\documentclass[twoside,leqno,twocolumn]{article}

\usepackage{ltexpprt}
\usepackage[hidelinks]{hyperref}

\usepackage{algorithmicx}
\usepackage[noend]{algpseudocode}
\usepackage{algorithm}
\usepackage{amsmath}
\usepackage{amssymb}
\usepackage{amsfonts}
\usepackage[capitalize]{cleveref}
\usepackage{subcaption}
\usepackage{multirow}
\usepackage{rotating}
\usepackage{graphicx}
\usepackage{color}
\usepackage{threeparttable}
\usepackage{placeins}
\usepackage[symbol]{footmisc}
\usepackage{todonotes}
\usepackage{xspace}
\usepackage{enumerate}

\newcommand{\set}[2]{\left\{\, \mathinner{#1}\vphantom{#2}\: \left|\: \vphantom{#1}\mathinner{#2} \right.\,\right\}}
\newcommand{\oneset}[1]{\left\{\, \mathinner{#1} \,\right\}}

\newcommand\ie{i.e\@., }

\newcommand{\Wlog}{W.\,l.\,o.\,g.\xspace}

\newcommand{\sse}{\subseteq}

\newcommand{\abs}[1]{\left|\mathinner{#1}\right|}

\newtheorem{observation}{Observation}[section]

\newcommand{\cD}{\mathcal{D}}

\newcommand{\cU}{\mathcal{U}}

\newcommand{\cX}{\mathcal{X}}

\newcommand{\universe}{U}
\newcommand{\pchild}[3]{{#1}\kern-1pt{+}{#2}{#3}}

\usepackage{tikz}

\usepackage{verbatim}
\usepackage{booktabs}
\newcommand{\armin}[2][]{\todo[color=green!30, #1]{#2}}
\newcommand{\florian}[2][]{\todo[color=orange!45, #1]{#2}}

\begin{document}

  \title{\Large Lower Bounds for Sorting 16, 17, and 18 Elements}

  \author{Florian Stober\thanks{FMI, University of Stuttgart, Germany} \and Armin Wei\ss$^*$
  }

  \date{}

  \maketitle







  \begin{abstract}
    \small\baselineskip=9pt
    It is a long-standing open question to determine the minimum number of comparisons $S(n)$ that suffice to sort an array of $n$ elements.
    Indeed, before this work, $S(n)$ has been known only for $n\leq 22$ with the exception of $n=16$, $17$, and $18$.

    In this work, we fill that gap by proving that sorting $n=16$, $17$, and $18$ elements requires $46$, $50$, and $54$ comparisons respectively.
    This fully determines $S(n)$ for these values and disproves a conjecture by Knuth that $S(16) = 45$.
    Moreover, we show that for sorting $28$ elements at least 99 comparisons are needed.

    We obtain our result via an exhaustive computer search which extends previous work by Wells (1965) and Peczarski (2002, 2004, 2007, 2012).
    Our progress is both based on advances in hardware and on novel algorithmic ideas such as applying a bidirectional search to this problem.

  \end{abstract}

  \section{Introduction}\label{sec:introduction}

  Sorting a set of elements is an important operation frequently performed by many computer programs.
  Consequently there exist a variety of sorting algorithms, each of which comes with its own advantages and disadvantages.
  A common measurement for the performance of a (comparison-based) sorting algorithm is the number of comparisons performed.
  The aim of this work is to prove lower bounds on the minimum number of comparisons required to sort a set of elements of a given size.

  We denote by $S(n)$ the minimum number of comparisons which suffice to sort $n$ elements in the worst case.
  It is well-known that $S(n) \ge \left\lceil \log_2 n! \right\rceil$.
  We call $C(n) = \left\lceil \log_2 n! \right\rceil$ the \emph{information theoretic lower bound}.
  The classical Mergesort algorithm comes already quite close to this lower bound~-- with a difference of only roughly $0.5n$ comparisons~-- see e.\,g.~\cite{Knuth}.
  The Ford-Johnson algorithm~\cite{FordJohnson} comes even closer to the lower bound, indeed being optimal for $n \le 11$.
  We denote by $F(n)$ the number of comparisons required by the Ford-Johnson algorithm to sort a set of $n$ elements.
  Obviously we have $C(n) \le S(n) \le F(n)$ and for $n \le 11$ even $C(n) = S(n) = F(n)$.
  For a detailed description of the Ford-Johnson algorithm we refer to Knuth's book~\cite[Section 5.3.1]{Knuth}, where the algorithm is called \emph{merge insertion}.

  For twenty years after the publication of the Ford-Johnson algorithm in 1959, it remained unclear whether the Ford-Johnson algorithm was optimal.
  Manacher~\cite{manacher1979ford} has been the first to show that this is not the case, and there are infinitely many $n$ for which $F(n) > S(n)$.
  Shortly after, it has been shown that the Ford-Johnson algorithm is not optimal for $n=47$~\cite{SCHULTEMONTING198119}.
  It is not known whether there is any $n < 47$ for which $F(n) > S(n)$.

  \paragraph{Lower Bounds via Computer Search.}\label{par:related-work}

  The first work to tackle the lower bound problem using an exhaustive computer search dates back to 1964, when
  Wells showed that $S(12) > C(12) = 29$~-- implying that $S(12) = F(12) = 30$~\cite{wells1964applications,wells2014elements}.

  His work is the foundation of many subsequent papers including the present.
  The algorithm explores the full search space of possible comparisons to sort some array applying some clever pruning techniques based on the number of linear extensions of intermediate partially ordered sets (posets)~-- for details see \cref{sec:basic_algorithm}.

  In 1994, Kasai, Sawato,  and Iwata succeeded to show that sorting 13 elements requires 34 comparisons~\cite{kasai1994thirty}.
  This implies that the Ford-Johnson algorithm is optimal for $n=13$.
  The computation took 230 hours, and the authors had to work within the limits of the computers available at the time.
  They had only 100\,MB of memory, which was not enough to store the entire search space.

  In the first decade of the new millennium Peczarski published several papers on the topic.
  Unaware of~\cite{kasai1994thirty}, he computed in 2002 that $S(13) = 34$~\cite{peczarski2002sorting}.
  This was followed by a journal version two years later, which additionally proved that the Ford-Johnson algorithm is optimal for sorting 14 and 22 elements~\cite{peczarski2004new}.
  The main difference between Wells' algorithm and the algorithms used by~\cite{kasai1994thirty,peczarski2002sorting,peczarski2004new} is that the latter sometimes need to check both outcomes of a comparison, whereas Wells' algorithm would only explore one outcome.
  Indeed, for $n=12$ it is sufficient to check only one outcome, but for $n=13$ and $n=14$ it is not.

  In 2007 a group of researchers ran Peczarski's algorithm on the Nankai Stars supercomputer, computing the values of $S(15)$ and $S(19)$~\cite{chinese2007}.
  Peczarski obtained the value for $S(15)$ independently some years earlier, but it was only published in 2007.
  Additionally he showed that for $n < 47$ there is no split-and-merge algorithm that beats the Ford-Johnson algorithm~\cite{peczarski2007ford}.
  By ``split-and-merge'' we mean an algorithm that splits the set into two disjoint subsets, sorts them independently using the Ford-Johnson algorithm and merges them afterwards.
  That class of algorithms is particularly interesting since Manacher used them to show that the Ford-Johnson algorithm is not optimal~\cite{manacher1979ford}.

  In an attempt to compute $S(16)$, in 2011 Peczarski suggested to use the algorithm for counting linear extensions presented in~\cite{LoofMB06}, which is significantly faster than the algorithm he had been using previously~\cite{peczarski2012towards}.
  Nevertheless, he did not succeed to compute $S(16)$.

  \subsection*{Contribution.}

  In this work we present a new algorithm to tackle the lower bound problem, following the footsteps of Wells and Peczarski.
  Using our algorithm we obtain the following:

  \begin{theorem}
    \label{thm:main}
    $S(16) = 46$, $S(17) = 50$, and $S(18) = 54$.
    Moreover, $S(28) \geq 99$.
  \end{theorem}

  Notice that in order to prove the theorem, we only have to show the lower bounds, \ie that 16 (resp.
  17, 18) elements are not sortable using 45 (resp.
  49, 53) comparisons.
  The upper bound is provided by the Ford-Johnson algorithm.
  Our result for $S(16)$ disproves a conjecture by Knuth~\cite[Section 5.3.1]{Knuth} that $S(16) < F(16)$.
  Moreover, we also confirm the previous results~\cite{wells1965applications,kasai1994thirty,peczarski2002sorting,peczarski2004new,peczarski2007ford,chinese2007}.

  Our result is made possible by advancements both in hardware and in software.
  In terms of hardware we rely on two Intel Xeon CPUs each having 12 cores (24 threads) and a total of 768\,GB of RAM (which is a great improvement compared to what Peczarski used~\cite{peczarski2012towards}, but less than the supercomputer in~\cite{chinese2007}).
  In terms of software we have several improvements compared to previous work.
  Here we want to highlight two points:
  \begin{itemize}
    \item We represent the partial orders occurring as intermediate results in a much more succinct data structure.
    By this easy improvement, our implementation needs only $ n(n-1)/2 + 5 $ \emph{bits} for each poset, whereas~\cite{peczarski2004new} uses $8(n + 2)$ \emph{bytes} per poset.
 For $n=16$ this means we reduce the space from 144 bytes in \cite{peczarski2004new} and still 68 bytes in \cite{peczarski2012towards} down  to only 16 bytes. 
    Moreover, with our implementation, storing the posets when running the search for $n=13$ would require about $30$\,MB storage (plus some index structures), while in~\cite{kasai1994thirty} even 100\,MB were not enough.
    \item Our most important algorithmic improvement is the use of a bidirectional search (see \cref{sec:bidirectional}).
    This allows us to reduce the size of the search space drastically.
    Indeed, for $n=13$, a pure forward search as in previous work generates 2032125 posets whereas our bidirectional search with properly chosen parameters  explores only 254574 posets~-- thus, the search space is 87\% smaller.
    This reduction certainly is even more dramatic for $n=16,17,18$ but hard to quantify precisely as we cannot store the full search space in memory.
    Reducing the size of the search space provides a speed-up in two ways: first, directly since we have to explore less posets, and, second, because of the reduced memory consumption less posets have to be discarded and explored several times.

    We want to point out that our approach for the bidirectional search differs from the common technique where searches from the start and endpoint are performed until the search spaces meet (see for example \cite{HolteFSSC17}). In our algorithm, we first perform a heavily pruned (and not complete) backward search~-- however, with no bound on the number of steps taken; after that we use the results as an advice for the forward search to allow for improved pruning. In our case, this novel approach results in the enormous reduction in search space size as detailed above.

  \end{itemize}

  \paragraph*{Outline.} In \cref{sec:basic_algorithm} we introduce the basic algorithm as described by Peczarski.
  After that we introduce our bidirectional approach in \cref{sec:bidirectional} followed by further details on our improvements and implementation in \cref{sec:details}.
  Finally, \cref{sec:experiments} presents some details on running times etc.\ of our experiments.

  \section{Preliminaries}\label{sec:preliminaries}

  \paragraph{Posets.}
  A partially ordered set, short \emph{poset}, is a pair $(\universe, R)$ where $\universe$ is a set and $R \sse \universe \times \universe$ is a partial order, \ie a reflexive, transitive and antisymmetric relation.
  Throughout this paper $\universe$ is finite.
  Thus, we can write $\universe = \{1, \dots, n\}$ where $n = |\universe|$.
  Two posets $P_A = (\universe, R_A)$ and $P_B = (\universe, R_B)$ are isomorphic if there is a bijective mapping $\varphi: \universe \to \universe$ such that $(u, v) \in R_A \iff (\varphi(u), \varphi(v)) \in R_B$ for all $u, v \in \universe$.
  By $P_0 = (\universe, R_0)$ we denote the unordered poset where $R_0 = \set{(u,u)}{u\in \universe}$.
  Usually, we write $\preceq$ instead of $R$ and $u\preceq v$ for $(u,v) \in R$.
  A poset is called a \emph{total order} if for all $u,v \in \universe$ either $u\preceq v$ or $v \preceq u$.

  The \emph{dual} of a poset is obtained by reversing the direction of all edges (\ie replacing $R$ by $\set{(v,u)}{(u,v) \in R}$).
  Note that a poset is sortable using $c$ comparisons if and only if its dual is sortable using $c$ comparisons.
  Therefore, according to~\cite{peczarski2004new}, we call two posets $P$ and $Q$ \emph{congruent} if $P$ and $Q$ are isomorphic or $P$ and the dual of $Q$ are isomorphic.

  A directed acyclic graph (\emph{DAG}) is a directed graph $(\universe,S)$ with an acyclic edge set $S\sse \universe \times \universe$.
  Its reflexive and transitive closure is a poset.
  More specifically, given some acyclic edge set $S \sse \universe \times \universe$, the poset generated by $S$ is defined as $P=(\universe, R)$ where $R$ is the smallest possible reflexive, transitive and antisymmetric relation with $S \sse R$.
  Given a poset $P = (\universe,\preceq)$, its Hasse diagram $(\universe, H)$ is given by the smallest subset $H \sse \universe \times \universe$ such that $P$ is the poset generated by $H$.
	For simplicity, we also call $H$ itself a \emph{Hasse diagram}.

  Given a poset $P = (\universe,R)$ and $u,v \in \universe$, we write $\pchild{P}{u}{v}$ for the poset generated by $R \cup \{(u,v)\}$, \ie the poset obtained from $P$ after adding the comparison $u \preceq v$.
  \begin{observation}
    \label{easyObservation}\
    \begin{itemize}
      \item If $u$ and $v$ are not comparable in $P$, then $(u,v)$ is always an edge in the Hasse diagram of $\pchild{P}{u}{v}$.
      \item The number of edges in the Hasse diagram of $\pchild{P}{u}{v}$ is at most the number of edges in the Hasse diagram of $P$ plus 1.
    \end{itemize}
  \end{observation}
  \newcommand{\eff}{E}
  \newcommand{\effTot}{E_{\mathrm{tot}}}
   \newcommand{\effThr}{E_{\mathrm{thr}}}

  \paragraph{Number of Linear Extensions and Efficiency.}
  Let $P=(\universe,R)$ be a poset.
  The number of linear extensions $e(P)$ is the number of total orders of $\universe$ that are compatible with $R$.
  Given a poset obtained from the unordered poset $P_0$ within $c$ comparisons, we define its \emph{efficiency} as
  \[\eff(P) = \frac{\abs{\universe}!}{2^c \cdot e(P)}.\]
  For details on the concept of efficiency we refer the reader to~\cite[Section 5.3.1]{Knuth}.
  Be aware that the efficiency of a poset implicitly depends on the number of comparisons $c$ by which the poset has been obtained.
  
  Assume we wish to sort $n = \abs{\universe}$ elements in $C$ comparisons (or prove that this is not possible). Then we denote with $\effTot = n! \cdot 2^{-C}$ the efficiency of a total order assuming that it has been obtained in $C$ comparisons from the unordered poset.
  
  Because of the following observation (which follows by the same argument as the information theoretic lower bound for sorting), the concept of efficiency is closely related to the question of whether a given poset might be sortable within $r$ comparisons (the second statement uses the implicit assumption $C = c+r$). 

  \begin{observation}
    \label{easyObservation2}
    If $e(P) > 2^r$, then $P$ is not sortable within $r$ comparisons.
    In other words, if $\eff(P) < \effTot$, then $P$ is not sortable within $r$ comparisons.
  \end{observation}

  The following crucial observation establishes that the efficiency is monotone in a certain sense:

  \begin{observation}
    \label{easyObservation3}
    Let $u,v \in \universe$.
    Then there is some $P' \in \{\pchild{P}{u}{v}, \pchild{P}{v}{u}\}$ such that $\eff(P') \leq \eff(P)$.
  \end{observation}

\begin{proof}
	Let $\{P',P''\} =\{\pchild{P}{u}{v}, \pchild{P}{v}{u}\}$ with $e(P') \geq e(P'')$. 	
	Then, we have $e(P') \geq e(P)/2$ and, hence,
$
		\eff(P') = \frac{\abs{\universe}!}{2^{c+1} \cdot e(P')}
				\leq \frac{\abs{\universe}!}{2^{c} \cdot e(P)} = \eff(P).
$
\end{proof}

  \paragraph{Sorting Algorithms.}
  We can view a sorting algorithm as a binary decision tree.
  Each node is labeled with a poset.
  The root node is the unordered poset $P_0$.
  The leaf nodes are posets with a total order.
  The sorting algorithm associates to each inner node $P$ a comparison $(u,v)$, meaning that from $P$ the next step is to compare $u$ with $v$.  
  The two children $\pchild{P}{u}{v}$ and $\pchild{P}{v}{u}$ correspond to the two possible outcomes of the comparison $u \preceq v$ and $u \succeq v$.
  The number of comparisons required by the algorithm (in the worst case) is the maximum length of a path from the root to some leaf.
  
    In terms of the search space described in \cref{fig:fw_search_tree}, a sorting algorithm can be viewed as choosing one outgoing edge (comparison) of every poset $P$.

  \section{The Basic Algorithm: Forward Search According to Peczarski}\label{sec:basic_algorithm}

  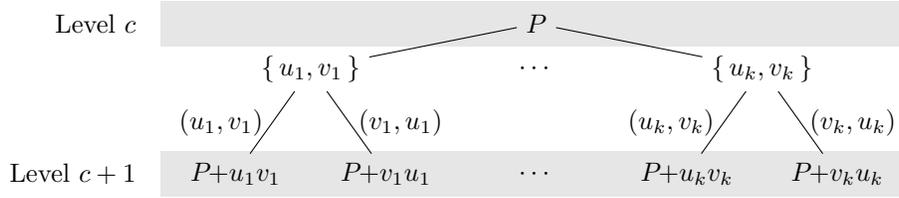
\begin{figure*}
    \centering
    \begin{tikzpicture}[xscale=2,yscale=-1]

      \fill [gray!20] (2.5,-0.3) rectangle (7.5,0.3);
      \fill [gray!20] (2.5,1.7) rectangle (7.5,2.3);

      \node (P) at (5,0) {$P$};
      \node (V1) at (3.5,0.6) {$\oneset{u_1, v_1}$};
      \node (Vk) at (6.5,0.6) {$\oneset{u_k, v_k}$};
      \node (1) at (3,2) {$\pchild{P}{u_1}{v_1}$};
      \node (2) at (4,2) {$\pchild{P}{v_1}{u_1}$};
      \node at (5,0.6) {$\cdots$};
      \node at (5,2) {$\cdots$};
      \node (3) at (6,2) {$\pchild{P}{u_k}{v_k}$};
      \node (4) at (7,2) {$\pchild{P}{v_k}{u_k}$};

      \draw (P) edge (V1)
      (P) edge (Vk)
      (V1) edge node[left] {$(u_1, v_1)$} (1)
      (V1) edge node[right] {$(v_1, u_1)$} (2)
      (Vk) edge node[left] {$(u_k, v_k)$} (3)
      (Vk) edge node[right] {$(v_k, u_k)$} (4);

      \node[anchor=east,align=right] at (2.4, 0) {Level $c$};
      \node[anchor=east,align=right] at (2.4, 2) {Level $c + 1$};
    \end{tikzpicture}

    \caption{Search space for the forward search.
    Level $c$ contains all posets obtainable from the unordered set with $c$ comparisons.
    For each pair of unrelated nodes $u$ and $v$, the poset $P$ has two children $\pchild{P}{u}{v}$ and $\pchild{P}{v}{u}$.}
    \label{fig:fw_search_tree}
  \end{figure*}

  In this section we describe the classical algorithm used by Peczarski~\cite{peczarski2002sorting,peczarski2004new,peczarski2007ford,peczarski2012towards} and others~\cite{kasai1994thirty,chinese2007}, which extends Wells' algorithm~\cite{wells1965applications}.
  The search starts with the unordered poset $P_0$ as root node.
  To show that $P_0$ is not sortable using at most $C$ comparisons, we show that there is no comparison such that both of the posets that result from it can be sorted using at most $C-1$ comparisons.
  In general, if we have a poset $P$ that has been obtained using $c$ comparisons from $P_0$, it is sortable in the remaining $C-c$ comparisons if and only if there is a comparison such that both of the successor posets can be sorted in $C-c-1$ comparisons.
  The predicate $S_c(P)$, denoting whether $P$ can be sorted using the remaining $C-c$ comparisons, can be expressed as follows:
  \begin{equation*}
    S_c(P) = \bigvee_{(u, v) \in \universe \times \universe} \bigl( \; S_{c+1}( \pchild{P}{u}{v} ) \, \land \,  S_{c+1}( \pchild{P}{v}{u}) \; \bigr)\label{eq:pred_sortable}
  \end{equation*}
  The search space for the forward search is depicted in \cref{fig:fw_search_tree}.
  The children of a poset always come in pairs, as each comparison has two possible outcomes, resulting in two child posets.

  For each comparison only one child is explored at first.
  Only if it turns out to be sortable, we check the other child.
  The intuition here is that most posets will turn out to be not sortable~-- so that way we save a lot of work.
  Indeed, in Wells' original algorithm, only one outcome of each comparison is checked~-- this actually suffices to show that $n = 12$ cannot be sorted in 29 comparisons.
  For 13 up to 19 elements Wells' original algorithm does not yield a lower bound.

  In order to keep the search space to a reasonable size the following points are crucial:
  \begin{itemize}
    \item We do not distinguish congruent (isomorphic or isomorphic to the dual) posets, as they can be sorted using the same number of comparisons.
    We give details on how to find congruent posets in \cref{subsec:poset-representation} and  \cref{subsec:hashing}. We store posets determined to be (un-)sortable for later use (\ie this is a kind of memoization approach).
    \item For each comparison the child with the larger number of linear extensions is checked first, as it is assumed to be harder to sort.
    \item Posets with $e > 2^r$ linear extensions cannot be sorted using $r$ comparisons (\cref{easyObservation2}).
  \end{itemize}

  \paragraph{Order of Exploration.}
  The order of exploration is not an essential feature of the algorithm: one could imagine a depth-first or breadth-first approach or some combination of them.
  Following Peczarski~\cite{peczarski2002sorting,peczarski2004new,peczarski2007ford,peczarski2012towards} our implementation traverses the search space in a breadth-first search manner, checking for each level, whether the encountered posets are sortable.
  However, notice that this is not a pure breadth-first search as only the possible candidates for comparisons (\ie pairs of elements to be compared) are explored in a breadth-first manner, while, as described above, the outcomes of the comparisons are explored in a depth-first manner  (\ie for a given poset $P$, the algorithm initially only explores $\pchild{P}{u}{v}$ or $\pchild{P}{u}{v}$ but not both) to reduce the size of the search space.
  Moreover, to reduce the memory usage, we even deviate slightly further from this breadth-first approach (like in~\cite{peczarski2004new,peczarski2007ford,peczarski2012towards})~-- for details see \cref{subsec:memory} below.

  \section{Improved Algorithm: Bidirectional Search}\label{sec:bidirectional}

  In this section we describe our main algorithmic improvement, namely the bidirectional search.
  We start by explaining the backward search which is the first step of the bidirectional search.%

  \subsection{Backward Search.}

  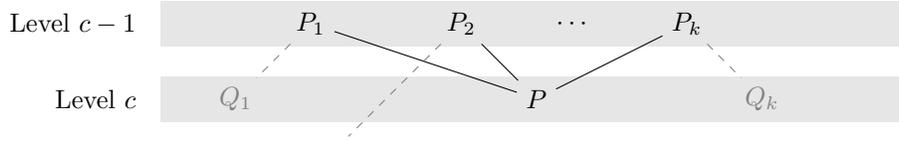
\begin{figure*}
    \centering
    \begin{tikzpicture}[xscale=2]

      \fill [gray!20] (2.5,-0.3) rectangle (7.5,0.3);
      \fill [gray!20] (2.5,0.7) rectangle (7.5,1.3);

      \node (P) at (5,0) {$P$};
      \node (1) at (3.5,1) {$P_1$};
      \node (2) at (4.5,1) {$P_2$};
      \node at (5.25,1) {$\cdots$};
      \node (3) at (6,1) {$P_{k}$};
      \node[color=gray] (Q1) at (3,0) {$Q_1$};
      \node[color=gray] (Qk) at (6.5,0) {$Q_k$};

      \draw (P) edge (1)
      (P) edge (2)
      (P) edge (3);
      \draw[dashed,color=gray]
      (1) edge (Q1)
      (2) edge (3.75,-0.5)
      (3) edge (Qk);

      \node[anchor=east,align=right] at (2.4, 0) {Level $c$};
      \node[anchor=east,align=right] at (2.4, 1) {Level $c - 1$};
    \end{tikzpicture}

    \caption{Search space for the backward search.
    Level $c$ contains all posets that can be sorted with $C - c$ comparisons.
    The poset $P$ can be obtained from its predecessors $\{P_1, \dots P_k\}$ by adding a suitable edge.
    For $P_i$ to be a predecessor, the poset $Q_i$ obtained by adding the reverse edge needs to be on a level $c' \ge c$.}
    \label{fig:bw_search_tree}
  \end{figure*}

  The idea behind the backward search is to start with a totally ordered poset and then enumerate all predecessors.
  Notice that there is only one congruence class of totally ordered posets.
  To keep notation consistent with the forward search, we will say that level $c$ contains all posets sortable in $C - c$ comparisons.
  We start with a totally ordered poset on level $C$.

  The predecessors of a poset $P$ on level $c$ are computed as follows.
  We begin by computing the set of potential predecessors of a poset.
  Those are posets from which we can obtain $P$ using one comparison (this does not mean that they are sortable using $C - c + 1$ comparisons yet).
  Each edge in the Hasse-Diagram of $P$ is a potential comparison by which $P$ can be obtained from a predecessor.
  However, be aware that even for one specific comparison $(u,v) \in \universe \times \universe$, there might be multiple potential predecessors $P'$ such that $P = \pchild{P'}{u}{v}$.
  Indeed, each combination of transitive edges that would be implied by the comparison is a valid option.
  In \cref{subsec:predecessorenumeration} we give details on how to find this set of potential predecessors.

  To check whether a potential predecessor $P'$ with $P = \pchild{P'}{u}{v}$ is an actual predecessor, we need to check whether $\pchild{P'}{v}{u}$ (\ie the poset obtained by the other outcome of the comparison $(u,v)$, hereafter called reverse edge poset) is sortable using at most $C - c$ comparisons.
  To be able to perform this test, it is important that a breadth first search is used as we need to look for $\pchild{P'}{v}{u}$ on level $c$ and all other levels up to $C$.
  \Cref{fig:bw_search_tree} depicts the search space for the backward search.
  We can check whether $n$ elements are sortable using $C$ comparisons by checking whether $P_0$ is contained in level 0.

  As an optimization we do not store posets on level $c$ with more than $c$ edges in their Hasse-Diagram.
  Let $B_c$ denote the set of (congruence classes of) posets obtained by the backward search on level $c$ when using this optimization.

  \begin{lemma}
    \label{lem:backwardSearchComplete}
    Let $P$ be a poset obtained by the forward search in $c$ comparisons and which is sortable in at most $C-c$ comparisons. Then
    $P \in \bigcup_{d \geq c} B_d$.
  \end{lemma}
  \begin{proof}
    Let $P'$ and $P''$ be the successors of $P$ in the forward search witnessing that $P$ is sortable in at most $C-c$ comparisons. In particular, $P'$ and $P''$ are sortable in $C-c-1$ comparisons. By induction, we have $P', P'' \in \bigcup_{d \geq c+1} B_d$. Let $\tilde c$ be maximal with  $P', P'' \in \bigcup_{d \geq \tilde c+1} B_d$. Then, without the optimization, we would have $P \in B_{\tilde c}$. Now, by \cref{easyObservation} the Hasse diagram of $P$ contains at most $c\leq \tilde c$ edges; thus, the optimization does not prevent $P \in B_{\tilde c}$.
  \end{proof}

  In particular, the backward search with the optimization is still complete: we can check whether $n$ elements are sortable using $C$ comparisons by checking whether $P_0$ is contained in $B_0$.
  However, even with that optimization, the search space for the backward search is much larger than for the forward search.
  Indeed, when proving that sorting 13 elements requires 34 comparisons, the backward search traverses $682\cdot 10^6$ posets compared to $1.8\cdot 10^6$ for the forward search.
  For further details see \cref{fig:count_13} and \cref{fig:eff_distr_n13_c18}.
  This makes the backward search on its own impractical.

  \begin{figure*}
    \centering
    \includegraphics[width=.8\textwidth]{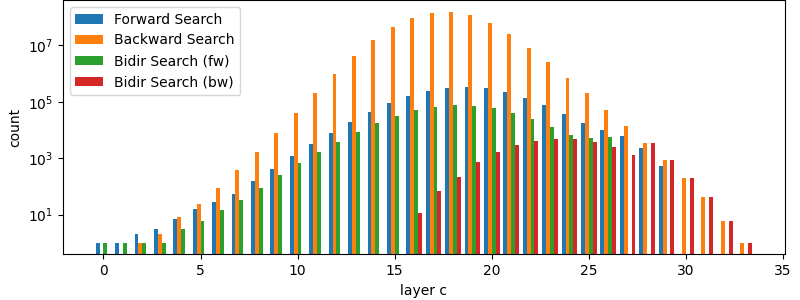}
    \caption{Comparison of the search space size for forward, backward and bidirectional search when searching an algorithm for sorting $13$ elements with $33$ comparisons. For the bidirectional search, an efficiency bandwidth of $0.05$ has been used. Be aware of the logarithmic scale of the $y$-axis.}
    \label{fig:count_13}
  \end{figure*}

  \begin{figure*}
    \centering
    \includegraphics[width=.8\textwidth]{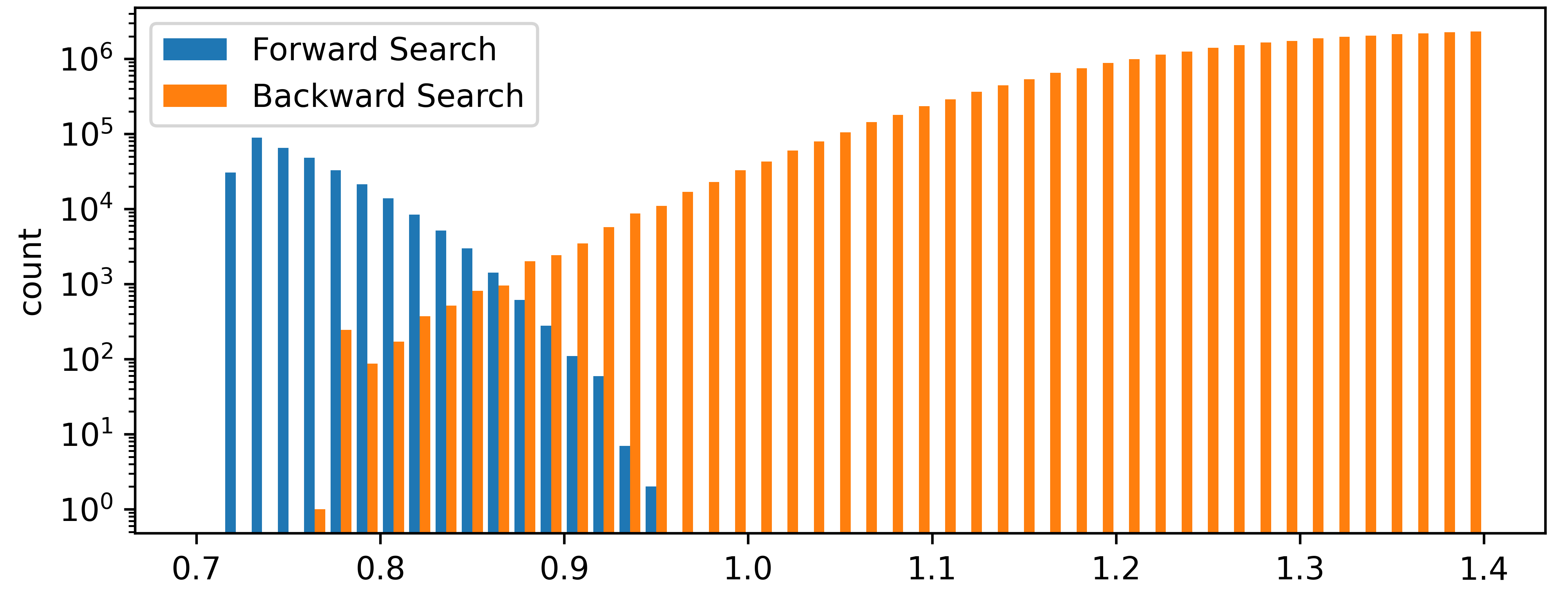}
    \caption{Histogram of efficiency of posets encountered by forward search on layer 18 compared with backward search, when searching an algorithm for sorting $13$ elements with $33$ comparisons.
    Be aware of the logarithmic scale of the $y$-axis.}
    \label{fig:eff_distr_n13_c18}
  \end{figure*}

  \subsection{Bidirectional Search.}

  \newcommand{\fullStart}{C_f}

  Comparing the posets traversed by the forward search with the posets encountered during a backward search, we observe that the intersection is rather small.
  This can be seen in \cref{fig:effdistr13}.
  The majority of posets encountered by the forward search has an efficiency smaller than $1$, whereas most posets encountered by the backward search have a much larger efficiency.
  We use this to our advantage, combining both algorithms.

  \paragraph{Partial Backward Search.}
  We modify the backward search as follows: we fix some constant $0 \leq \fullStart \leq C$. 
  After running the backward search for the layers layers $c \geq \fullStart$ as usual (called the \emph{full layers}), for the remaining layers (with $c < \fullStart$) we only store posets with an efficiency bounded by a user-specified threshold $\effThr$ with $\effTot \leq \effThr< 2\cdot \effTot$ where $\effTot$ is the efficiency of a total order obtained after $C$ comparisons (for the precise values of $\fullStart$ and $\effThr$, see \cref{sec:experiments}).
    We call $\effThr - \effTot$ the \emph{efficiency bandwidth}.
  The smaller $\effThr$, the less posets need to be searched and the faster the backward search.
  Note that, if $\effThr < 1$, the search will not find the unordered poset $P_0$ on level $0$.
  In particular, with our modification the backward search on its own is no longer sufficient to decide if there is a sorting algorithm.

  Another issue with the partial backward search is that we can obtain a potential predecessor whose reverse edge poset is above the efficiency threshold.
  In that case we cannot decide whether it is a predecessor or not.
  We add a flag indicating that it is unknown whether the poset is sortable using the remaining number of comparisons.
  \begin{figure*}
    \centering
    \includegraphics[width=.8\textwidth]{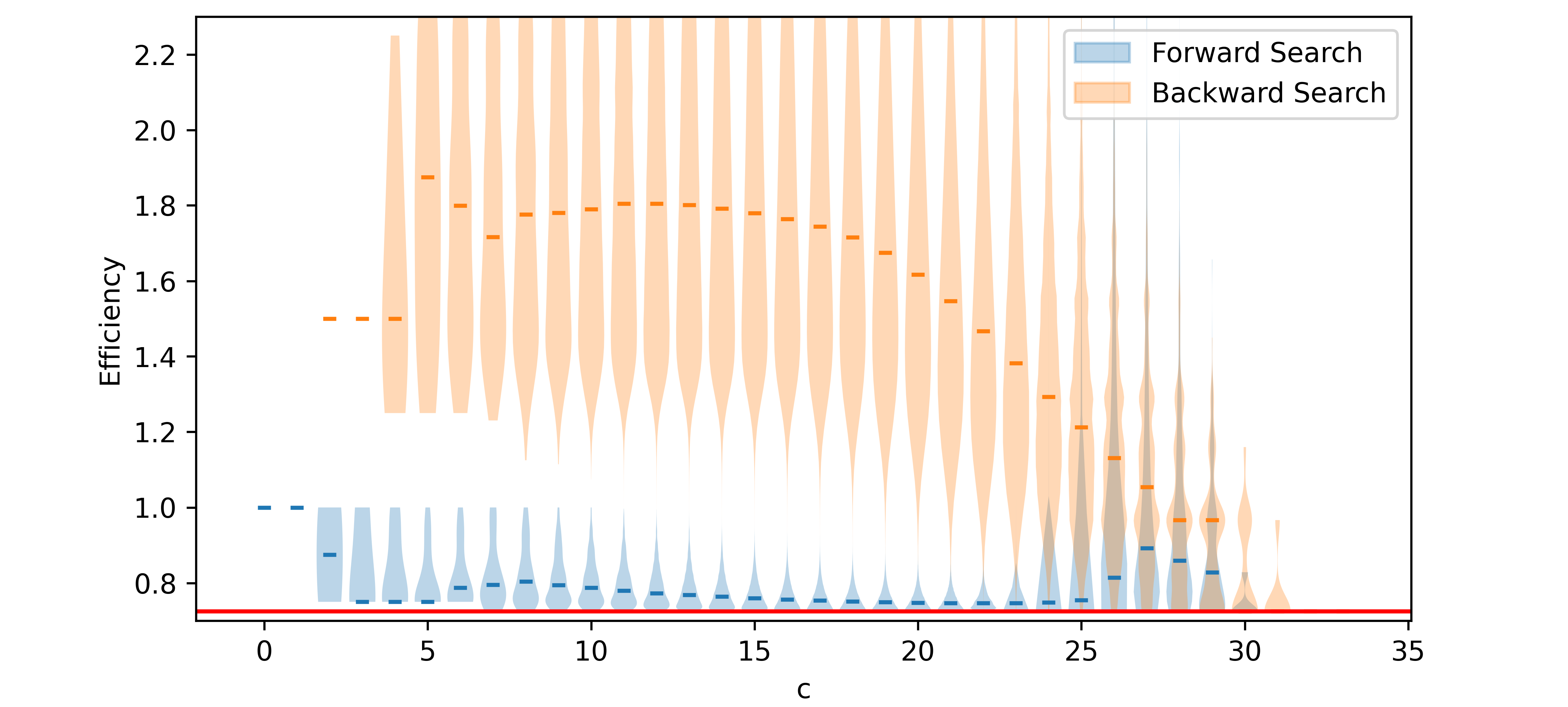}
    \caption{Comparison of efficiency distribution of posets encountered by the forward search vs.\ backward search when searching for an algorithm to sort $13$ elements using $33$ comparisons.
    Note that for most layers the backward search finds a small number of posets with efficiency less than 1, despite this not being visible in the visualization.}
    \label{fig:effdistr13}
  \end{figure*}
  Still, we have the following partial completeness result, where $B'_c$ denotes the set of (congruence classes of) posets obtained by the partial backward search on level $c$.
  \begin{lemma}
    \label{lem:partialBackwardSearch}
    Let $P$ be a poset obtained by the forward search in $c$ comparisons and sortable in at most $ C-c$ comparisons.
    Then
    \begin{enumerate}
      \item If $c \geq \fullStart$, then $P \in \bigcup_{d \geq c} B'_d$.
      \item If $\eff(P) \leq \effThr$, then $P \in B'_c$.
    \end{enumerate}
  \end{lemma}
  \begin{proof}
    The first point is due to \cref{lem:backwardSearchComplete}.
    We prove the second point by induction on $c$.
    For $c = C$, the poset $P$ is sortable in zero comparisons and, hence, by definition $P \in B_C'$.
    
    For $c < C$ we have $\eff(P) \leq \effThr < 2\cdot \effTot$; hence, $P$ is not a total order.
    Let $P'$ and $P''$ be the successors of $P$ in the forward search witnessing that $P$ is sortable in at most $C-c$ comparisons.
    Without loss of generality $e(P') \geq e(P'')$.
    Therefore, we have $e(P') \geq e(P)/2$ and, thus, $\eff(P') \leq \eff(P) \leq \effThr$ (like in \cref{easyObservation3}). By induction, we have $P' \in B'_{c+1}$.
    As $\eff(P) \leq \effThr$, also $P$ has been taken into $B'_c$ (either marked as ``unknown'' or not).
  \end{proof}

  \paragraph{Forward Search with Advice.}
  After finishing the partial backward search with the threshold, we perform the usual forward search.
  Now, we use the information obtained by the backward search to speed it up:
  whenever the forward search encounters a poset with an efficiency at most the threshold $\effThr$, it is looked up in the backward search results.
  If it is not contained, then by \cref{lem:partialBackwardSearch} it is not sortable within the remaining number of comparisons.
  If it is contained in the backward search results and not marked as ``unknown'', then it is sortable.
  If it is contained in the backward search results and marked as ``unknown'', then the forward search has to proceed further.
  As the majority of posets encountered by the forward search have a small efficiency, we can decide those posets using the results from the backward search.
  Thus, the size of the search space is reduced significantly~-- see \cref{fig:count_13}.

  \section{Implementation Details}\label{sec:details}

  Our algorithm begins with a partial backward search.
  Once complete, the results are loaded into memory and the forward search is started with the poset $P_0$ on level 0.
  The forward search proceeds level by level.
  On each level the algorithm goes through the following phases:
  \begin{description}
    \item[Phase 0] The posets to be explored are sorted according to their number of linear extensions.
    \item[Phase 1] The children of each poset are generated and one child of each comparison is added to the active posets of the child layer.
    If there are too many (more than a preset limit), we do this in several chunks and the following phases are executed for each chunk separately.
    The algorithm continues with Phase 0 of the child layer for the current chunk.
    \item[Phase 2] For each sortable poset on the child layer we add the other outcome of the associated comparison to the active posets of the child layer.
    The algorithm continues with Phase 0 of the child layer (for the newly added posets).
    \item[Phase 3] The posets are marked as sortable resp.\ not sortable depending on the results in the child layer.
    The active posets on the child layer are moved to the separate storage for old posets.
    If there are more chunks of child posets to be explored, the algorithm continues with Phase 1 for the next chunk;
    otherwise,
    the algorithm continues with the parent layer.
  \end{description}

  \subsection{Memory Management.}\label{subsec:memory}

  We use SSDs and hard disk drives to increase the amount of posets we can store.
  Let us describe the different phases of the algorithm separately.

  \paragraph{Backward Search.}
  We keep only two layers of the search tree in memory.
  The layer that is currently being processed, and the layer above that is being created.
  The remaining layers are stored on the hard drive.

  \paragraph{Forward Search.}
  We split the memory into two parts.
  One for active posets, which are the posets on the layer being processed and their parents (all the way up to $P_0$ on layer 0~-- these are exactly those posets which the algorithm has already started to explore but which are not yet known to be sortable or not sortable).
  As described above, we limit the maximum number of posets per layer.
  If there are more than the limit, we split them into different chunks (the size of the limit) and continue the search for each chunk separately.
  We keep them in an array of fixed size, three times the limit per layer.
  Posets on parent layers that do not fit in the array are stored on the hard drive.

  The second part of memory is for old generation posets.
  These are posets encountered when processing the layer in a previous iteration.
  We allocate a hashmap of fixed size for each layer to store them.
  The posets themselves are stored on an SSD drive and we keep only part of their hash in the main memory.
  Since most negative lookups only access the main memory, this results in a good trade-off between memory usage and lookup performance under the assumption that most lookups will have a negative result (which, indeed, is the case).
  When the map fills up we prefer to store sortable posets over unsortable ones, as they are potentially harder to compute.
  Note that we do not store all posets due to the fixed size of the map and a limited amount of main memory and SSD storage.
  Therefore, some posets have to be explored multiple times.

  \paragraph{Bidirectional Search.}
  In addition to the techniques presented above, when performing the forward search, the entire result of the partial backward search is loaded into main memory.

  \subsection{Parallelization.}

  Exploiting parallelism is mostly straight-forward.
  The presented algorithms explore the search space in a breadth-first search manner.
  We process each layer in parallel, by splitting the posets into several chucks and distributing them to the available threads such that the work is split evenly.
  The biggest challenge is the data structure that stores the posets of the next level.
  Those posets are generated during the parallel phase.
  Therefore, we need a hashmap that supports parallel reads and writes, while minimizing contention.
  We solve this issue using a two-level hash map.
  The first level has a constant size, is immutable, and, therefore does not need any safeguards for concurrent access.
  Each map on the second level is guarded by a separate mutex.
  We choose the first layer sufficiently large, to keep the chance of two threads accessing the same map at the same time low, thus reducing contention.

  \subsection{Poset Representation.}\label{subsec:poset-representation}

  A poset is a DAG\@.
  When storing it in memory we omit transitive edges, \ie we store its Hasse diagram.
  The DAG is stored as an adjacency matrix.
  Storing an adjacency matrix with $n$ nodes usually requires $n^2$ bits of memory.
  To reduce that number, we sort the nodes topologically, \ie we reorder them such that for each edge $(i, j)$ we have $i < j$.
  This is possible because the graph is a DAG\@.
  By doing so, we only need to store an upper triangular matrix meaning that we only need $n(n-1)/2$ bits.

  \paragraph{Canonical Poset Representations.}\label{par:canonicalrep}
  Isomorphic or congruent graphs can have different adjacency matrices.
  To mitigate this, we compute a canonical representation for each poset.
  For that purpose we use a technique known as color refinement (see~\cite{KieferM20} for some recent results and further references).
  The basic algorithm can be described as follows: start by assigning the same color to every vertex.
  While there are two vertices of the same color for which the number of neighbours of one specific color disagrees, give those two vertices different colors.
  We implement this using a hashing-based approach:
  Instead of explicitly counting neighbours of a specific color, we compute for each node a hash value based on the old colors of the node and its neighbours.
  The hash value is the new color.
  Here we give the pseudo-code:
  \begin{algorithmic}[1]
    \For{$v \in \universe$}
      \State degHash[$v$]\!\! $\gets$\!\! \textsc{hashFn}($\operatorname{in-degree}(v)$,
      \State\qquad\qquad\qquad\qquad\quad\,\, $\operatorname{out-degree}(v)$)
      \State hash[$v$] $\gets$ degHash[$v$]
    \EndFor
    \For{number of rounds}
      \For{$v \in \universe$}
        \State sum[$v$] $\gets$   hash[$v$]$ +\, \smash{\sum_{(u,v) \in E }}$ hash[$u$]
        \State\qquad\qquad\qquad\quad\,\, $ +\, \smash{\sum_{(v,u) \in E }}$  hash[$u$]
      \EndFor
      \For{$v \in \universe$}
        \State hash[$v$] $\gets$  \textsc{hashFn}(sum[$v$], degHash[$v$])
      \EndFor
    \EndFor
  \end{algorithmic}

  Two nodes having a different set of neighbours are likely to get a different hash value.
  There is a small chance that our hash-based algorithm leads to a worse coloring than an exact algorithm.
  However, we happily accept that, to have a faster algorithm.

  The canonical representation of the poset is created by ordering the nodes of the DAG based on the colors obtained by the color refinement algorithm.
  In some cases multiple vertices have the same color meaning that we do not obtain a unique order.
  This can have two reasons: first, the DAG could have a non-trivial automorphism (note that if two nodes $u$, $v$ are mapped to each other under some automorphism, then clearly they will have the same color).
  Second, our color refinement algorithm simply could have failed to tell the two nodes apart.
  We added two easy heuristics to detect and deal with the first case:
  \begin{itemize}
    \item If there are two vertices $u$, $v$ with the same color but no other vertex shares this color, we exchange the rows and columns corresponding to $u$ and $v$ in the adjacency matrix and test whether it is bitwise equal to the original one.
    If this is the case, we know that $u$ and $v$ are mapped to each other by some automorphism and for the canonical ordering it does not matter whether $u$ or $v$ comes first as the adjacency matrix is the same in any case.
    \item If there are $k > 2$ vertices $u_1, \dots, u_k$ with the same color but no other vertex shares this color, we test whether the full symmetric group $S_k$ acts as automorphisms of our DAG.
    As $S_k$ is generated by a transposition and a $k$-cycle, we only need to test whether the adjacency matrix is invariant under these two transformations.
    This is done as before.
  \end{itemize}
  If these two heuristics do not tell us that we are in the case of a non-trivial automorphism, we assume that we are in the second case, \ie that the color refinement algorithm has failed to produce a canonical representation.

  Because of this possibility, we add two additional bits to the graph representation.
  If we could not compute a canonical representation, we set a bit to indicate that the representation is not unique, and handle this poset differently when computing hash values and testing for congruence/isomorphism.
  The other bit is set in the case that there is no unique representation with respect to the dual of the graph (we do not discuss this in detail as this issue is similar but rather simpler than the one of non-trivial automorphisms).

  \subsection{Hashing and Isomorphism Test.}\label{subsec:hashing}

  For posets for which we have a canonical representation (which we know by looking at the two bits in the poset representation introduced in \cref{par:canonicalrep}), we compute the hash value from the adjacency matrix.
  Congruence or equality of posets is tested by comparing the adjacency matrices bitwise.

  For the remaining posets the hash value is computed from the graph structure, using the hash based color refinement algorithm presented in \cref{subsec:poset-representation}.
  After computing the colors for each vertex, we compute a single hash of all colors that is independent of the vertex order.
  Congruence testing is done using a graph isomorphism test (using \texttt{boost::vf2\_graph\_iso}), potentially also involving the dual.
  That is an expensive computation, but it is only required for few posets.
  When searching for an algorithm to sort 17 elements using 49 comparisons, about 1.4\% of all isomorphism tests were carried out using the Boost graph isomorphism test.
  The remaining 98.6\% were done using a bitwise comparison of the canonical representation.

  \subsection{Heuristics for Search Space Reduction.}\label{subsec:searchspacereduction}

  \paragraph{Sort Posets.}
  When processing posets in the forward search, we order them by their number of linear extensions.
  The order influences how many posets are explored on the child layer.
  In particular, for each comparison it is sufficient to explore one child.
  If one child has already been selected by a different comparison explored before, then the other child will not be selected additionally, even if the selected child is the one with fewer linear extensions.
  In our experiments, ordering the posets on the parent layer by their number of linear extensions, reduced the size of the child layer.

  \paragraph{Bound on the Number of Pairs.}
  A pair is a connected component of a poset consisting of exactly two vertices.
  A comparison creating a pair is independent of comparisons not involving the two vertices of the pair.
  Considering a sequence of comparisons, a comparison creating a pair can be swapped with its predecessor resp.\ successor comparison if they are independent.
  We use this to our advantage, by only considering sequences where comparisons creating a pair appear as late as possible.
  For the forward search this means that, if we create a pair, then either the next comparison involves a vertex of the pair or another pair is created and the comparison thereafter involves one vertex of each pair.
  For the backward search it is even simpler.
  Whenever a poset contains a pair, we only consider the predecessor without the pair that will lead to the poset by creating the pair.

 \begin{figure*}
	\centering
	\begin{subfigure}[c]{0.25\textwidth}
		\centering
		\begin{tikzpicture}[xscale=0.4,yscale=0.4]
			\coordinate (1) at (-0.6, 0);
			\coordinate (2) at (2.6, 0);
			\coordinate (3) at (1, 2);
			\coordinate (4) at (1, 4);
			\coordinate (5) at (-0.6, 6);
			\coordinate (6) at (2.6, 6);
			\coordinate (7) at (1, 8);
			\fill (1) circle (5pt) node[left]{1};
			\fill (2) circle (5pt) node[left]{2};
			\fill (3) circle (5pt) node[left]{3};
			\fill (4) circle (5pt) node[left]{4};
			\fill (5) circle (5pt) node[left]{5};
			\fill (6) circle (5pt) node[left]{6};
			\fill (7) circle (5pt) node[left]{7};
			
			\draw[] (1) -- (3);
			\draw[] (2) -- (3);
			\draw[,thick] (3) -- (4);
			\draw[] (4) -- (5);
			\draw[] (4) -- (6);
			\draw[] (5) -- (7);
			\draw[] (6) -- (7);
		\end{tikzpicture}
		\subcaption{A poset.}
		
	\end{subfigure}
	\begin{subfigure}[c]{0.25\textwidth}
		\centering
		\begin{tikzpicture}[xscale=0.4,yscale=0.4]
			\coordinate (1) at (-0.6, 0);
			\coordinate (2) at (2.6, 0);
			\coordinate (3) at (1, 2);
			\coordinate (4) at (1, 4);
			\coordinate (5) at (-0.6, 6);
			\coordinate (6) at (2.6, 6);
			\coordinate (7) at (1, 8);
			
			\draw[] (1) -- (3);
			\draw[] (2) -- (3);
			\draw[] (4) -- (5);
			\draw[] (4) -- (6);
			\draw[] (5) -- (7);
			\draw[] (6) -- (7);
			
			\draw[gray] (1) -- (4);
			\draw[gray] (1) -- (5);
			\draw[gray] (1) -- (6);
			\draw[gray] (1) -- (7);
			\draw[gray] (2) -- (4);
			\draw[gray] (2) -- (5);
			\draw[gray] (2) -- (6);
			\draw[gray] (2) -- (7);
			\draw[gray] (3) -- (5);
			\draw[gray] (3) -- (6);
			\draw[gray] (3) to[out=110,in=257] (7);
			
			\fill (1) circle (5pt);
			\fill (2) circle (5pt);
			\fill (3) circle (5pt);
			\fill (4) circle (5pt);
			\fill (5) circle (5pt);
			\fill (6) circle (5pt);
			\fill (7) circle (5pt);
		\end{tikzpicture}
		\subcaption{Poset without the selected comparison.}
		
	\end{subfigure}
	\begin{subfigure}[c]{0.4\textwidth}
		\small
		\begin{tabular}{|lr|l|l|l|l|l|l|l|l|}
			\cline{3-9}
			\multicolumn{1}{l}{} & & \multicolumn{7}{c|}{target} \\ \cline{3-9}
			\multicolumn{1}{l}{} &                         & 1 & 2 & 3 & 4              & 5 & 6 & 7               \\ \hline
			\multirow{7}{*}{\begin{sideways}
					source
			\end{sideways}}
			& \multicolumn{1}{|r|}{1} &   & 0 & 1 & ?              & ? & ? & ?               \\ \cline{2-9}
			& \multicolumn{1}{|r|}{2} &   &   & 1 & ?              & ? & ? & ?               \\ \cline{2-9}
			& \multicolumn{1}{|r|}{3} &   &   &   & {\color{red}0} & ? & ? & ?               \\ \cline{2-9}
			& \multicolumn{1}{|r|}{4} &   &   &   &                & 1 & 1 & {\color{gray}1} \\ \cline{2-9}
			& \multicolumn{1}{|r|}{5} &   &   &   &                &   & 0 & 1               \\ \cline{2-9}
			& \multicolumn{1}{|r|}{6} &   &   &   &                &   &   & 1               \\ \cline{2-9}
			& \multicolumn{1}{|r|}{7} &   &   &   &                &   &   &                 \\ \hline
		\end{tabular}
		
		\subcaption{Adjacency matrix of potential predecessors.}
	\end{subfigure}
	\caption{Predecessor enumeration example.}
	\label{fig:predecessor_enumeration}
\end{figure*}
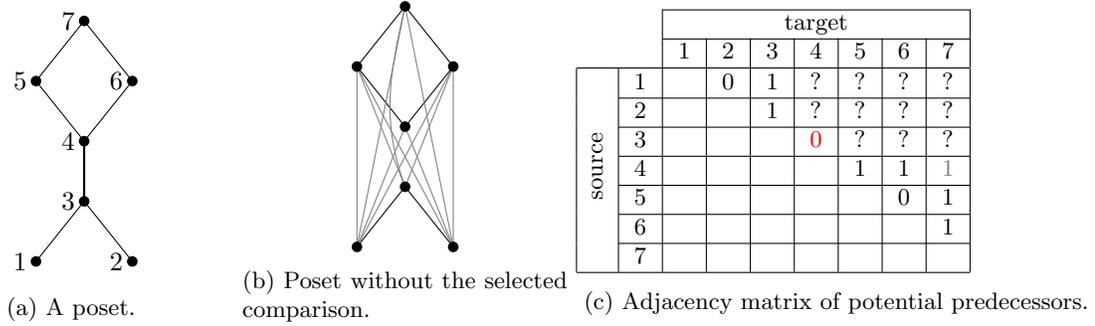

  \subsection{Computation of Linear Extensions.}

  For computing the number of linear extensions, we implemented the algorithm proposed by Peczarski in~\cite{peczarski2012towards}, which has been introduced in~\cite{LoofMB06}.
  As outlined in~\cite{peczarski2012towards} this is a big improvement compared to previous algorithms.
  Let $ P = (\universe, \preceq)$ be a poset.
  A subset $D\sse \universe$ is called a \emph{down-set} if $x \in D$ and $y \preceq x$ implies that $y \in D$.
  We denote the restriction of $\preceq$ to $D \times D$ also with $\preceq$ and write $P_D = (D,\preceq)$.
  We write $\cD(P)$ for the set of down-sets of $P$.
  By~\cite{LoofMB06} we have
  \begin{align}
    e(P_D) = \sum_{\substack{u \in D\\D \setminus\{u\} \in \cD(P_D)}} e(P_{D \setminus\{u\}}).\label{eq:recursion_downsets}
  \end{align}
  The same way we can define an \emph{up-set} $V$ and compute its number of linear extensions.
  We write $\cU(P)$ for the set of up-sets of $P$.
  Note that $V$ is an up-set if and only if $\universe \setminus V$ is a down-set.
  \newcommand{\DU}{\operatorname{DU}}%
  Let
  \vspace{-1mm}
  \begin{align*}
    \DU(u,v)  &= \bigl\{(D,V) \in \cD(P) \times \cU(P) \;\big| \;  v \in V,\\
    &\qquad \qquad \{u\} = \universe \setminus( D \cup V),\, D \cap V = \emptyset\bigr\}.
  \end{align*}
  \vspace{-1mm}
  Then
  \begin{align}
    e( \pchild{P}{u}{v}) = \sum_{(D,V) \in \DU(u,v) } e(P_{D })\cdot e(P_{V}).\label{eq:tablefill}
  \end{align}
  In~\cite{peczarski2012towards} the recurrence~\eqref{eq:recursion_downsets} is used to compute the $e(P_D)$ for all down-sets $D$ and $e(P_V)$ for all up-sets $V$ and then~\eqref{eq:tablefill} is used to compute the number of linear extensions resulting after each possible comparison.
  We follow the same approach.
  In order to speed up the computation of~\eqref{eq:tablefill}, we use AVX2 intrinsics.

  \paragraph*{Avoiding Many Singletons.}
  We call an element in a poset without incident edges in the DAG a \emph{singleton}.
  If a poset contains more than one singleton, we remove all except one of them before computing the linear extensions.
  We can do this safely because, when comparing a singleton with a non-singleton, for counting linear extensions it does not matter which singleton we actually compare.
  Moreover, when comparing two singletons, the number of linear extensions is simply half the number of linear extensions of the original poset.

  \paragraph*{Reducing Memory Requirements and Improved Cache Behavior.}
  In~\cite{peczarski2012towards} Peczarski uses a table of size $2^n$ for storing the numbers of linear extensions of the down-sets.
  For $n=16$, $17$, and $18$, we follow the same approach.
  However, for $n=28$ we use a smaller table since otherwise these tables would consume more than 400\,GB of memory (for 48 threads together).

  To bound the number of down-sets we have to store, we use a heuristic argument based on the following observation:

  \begin{observation}
    Let $(\universe, \preceq)$ be a partial order and $M \sse \binom{\universe}{2}$ be a perfect matching (\ie for all $u \in \universe$ there is some $e \in M$ with $u \in e$ and $e \cap f = \emptyset $ for all $e\neq f \in M$).
    If $M  \sse {\preceq}$, then the number of down-sets of $(\universe, \preceq)$ is bounded by $\sqrt{3}^{\abs{\universe}}$.
  \end{observation}
  \begin{proof}
    We start with the following observation: if posets $P$ and $Q$ have $k_P$ (resp.\ $k_Q$) down-sets, then the disjoint union $P \cup Q$ has $k_P \cdot k_Q$ down-sets.
    Hence, it suffices to observe that a totally ordered two-element poset (what we call a pair) has $3$ down-sets.
  \end{proof}
  In general, a poset will not have a perfect matching.
  Still the bound $\sqrt{3}^{\abs{\universe}}$ turns out to be a good heuristic for posets appearing during our algorithm and which have at most one singleton element.
  As we do not have a proof for our bound $\sqrt{3}^{\abs{\universe}}$, we allocate space for $\sqrt{3}^{\abs{\universe}+2}$ down-sets and include a check whether this space actually suffices.
  In \cref{tab:downsets} it can be seen that $\sqrt{3}^{\abs{\universe}}$ is actually quite good estimate for the number of down-sets.
  \begin{table*}
    \small
    \caption{The maximum number of down-sets of $n$ element posets occurring in the bidirectional search.}\label{tab:downsets}
    \centering\begin{tabular}{|c|c|c|c|c|c|c|c|c|c|c|c|}
                \hline
                $n$                                     & 11  & 12  & 13  & 14   & 15   & 16   & 17   & 18   & 19    & 22    & 28      \\
                \hline
                $\sqrt{3}^{n}$ & 421
                & 729 & 1263 & 2187 & 3788 & 6561 & 11364 & 19683 & 34092 & 177147 & 4782969 \\
                \hline
                \parbox{2cm}{max.\ number of down-sets} & 342 & 441 & 882 & 1554 & 2565 & 6156 & 9236 & 9702 & 15246 & 51102 & 1058841 \\
                \hline
    \end{tabular}
  \end{table*}

  In order to use the smaller table for down-sets, we also have to store the down-sets themselves as they cannot be read from the index as in the case of a table of size $2^n$.
  We store the down-sets in lexicographical order (where a down-set is identified with a bit-string in a natural way).

  \subsection{Predecessor Enumeration.}\label{subsec:predecessorenumeration}

  Predecessor enumeration refers to computing all (potential) predecessors of a poset, which is required for the backward search algorithm.
  We explain predecessor enumeration with the help of the example in \cref{fig:predecessor_enumeration}.
  Given a poset $P = (\universe, \preceq)$ the outer most loop of the algorithm iterates over all edges of the Hasse diagram.
  For each edge, the algorithm needs to enumerate all posets from which $P$ can be obtained by adding the edge.
  Considering an edge $(u, v)$, in the example $(3, 4)$, we look at the set of transitive edges implied by $(u, v)$:
  \[
    S = \{(x, y) \mid x \preceq u \land v \preceq y \land (x \neq u \lor y \neq v)\}
  \]
  Every poset obtained by removing $(u, v)$ and a subset of $S$ from $\preceq$ is a potential predecessor, and every potential predecessor can be obtained this way.
  Our code works with Hasse diagrams representing partial orders.
  Let $\preceq_H$ be the edges of the Hasse diagram of $P$.
  We use the function $\textsc{enumeratePredecessors}(\{(u,v)\}, \preceq_H)$ depicted below to compute the set of all potential predecessors
  (here $R^*$ denotes the reflexive and transitive closure of $R$).
  \begin{algorithmic}[1]
    \Function{enumeratePredecessors}{$E$, $R$}
      \If{$E = \emptyset$}
        \State return $\{ R \}$
      \EndIf
      \State $(u, v) \gets$ \textsc{pop}($E$)
      \State $A \gets $ \textsc{enumeratePredecessors}($E$, $R$)
      \State $R_1 \gets R \setminus \{ (u, v) \}$
      \State $S' = \{(x, v) \mid (x, u) \in R \land (x, v) \notin R_1^*\}$
      \State\qquad\qquad $\cup \{(u, y) \mid (v, y) \in R \land (u, y) \notin R_1^*\}$
      \State $B \gets $ \textsc{enumeratePredecessors}($E \cup S'$,
      \State\qquad\qquad $R_1 \cup S'$)
      \State return $A \cup B$
    \EndFunction
  \end{algorithmic}

  \begin{lemma}
    \label{lem:51}
    Let $R$ be a Hasse-Diagram and $(u, v) \in R$.
    Let $R_1 = R \setminus \{ (u, v) \}$ and
    \begin{align*}
      R_2 = R_1 &\cup \{(x, v) \mid (x, u) \in R \land (x, v) \notin R_1^*\} \\
      &\cup \{(u, y) \mid (v, y) \in R \land (u, y) \notin R_1^*\}\text{.}
    \end{align*}
    Let $Q$ be a Hasse-Diagram.
    Then
    \begin{enumerate}[(1)]
      \item $R_1$ and $R_2$ are Hasse-Diagrams,
      \item\label{pointTwo} $R_2^* = R^* \setminus \{ (u,v) \}$ and
      \item\label{pointThree} $(Q \cup \{(u,v)\})^* = R^*$ if and only if $R_1^* \subseteq Q^* \subseteq R^*$
    \end{enumerate}
  \end{lemma}

  \begin{proof}
  	$R_1$ being a Hasse-Diagram is obvious.
  	To see that $R_2$ is a Hasse-Diagram, observe that we only add edges $(x, u)$ (resp.\ $(v, y)$) for which there is no path $x$ to $u$ (resp.\ $v$ to $y$) in $R_1$.
  	
  	For the second statement we show that $R_2^* = R^* \setminus \{ (u,v) \}$.
  	It is obvious from the definition that $R_2^* \subseteq R^* \setminus \{ (u,v) \}$.
  	Now let $(x, y) \in R^* \setminus \{ (u,v) \}$.
  	Then, there is a directed path $\Phi$ from $x$ to $y$ in $R$.
  	If $\Phi$ does not use the edge $(u,v)$, then it exists in $R_2$ and, hence, $(x,y) \in R_2^*$.
  	Thus, we consider the case $\Phi = (x_0, \dots x_i, u, v, y_1, \dots, y_j )$ with  $x = x_0$ and $y_j = y$.
  	\Wlog $x \neq u$.
  	By definition there is an edge $(x_i, v)$ in $R_2$.
  	Hence, the path $\Phi' = (x_0, \dots x_i, v, y_1, \dots, y_j)$ exists in $R_2$ and $(x,y) \in R_2^*$.
  	
  	For the third statement assume that $(Q \cup \{(u,v)\})^* = R^*$.
  	It is obvious that $Q^* \subseteq R^*$.
  	We show that $R_1 \subseteq Q$.
  	Assume there is some $(x,y) \in R_1 \setminus Q$.
  	Then the only path from $x$ to $y$ in $R$ is the edge $(x, y)$ (otherwise $R$ would not be a Hasse-Diagram).
  	Hence, there would be no edge $(x,y)$ in $(Q \cup \{(u,v)\})^* \subseteq R^*$ contradicting $(Q \cup \{(u,v)\})^* = R^*$.
  	To see that the other direction holds, observe that $(R_1 \cup \{(u,v)\})^* = R^*$.
  \end{proof}

  \begin{lemma}
  	Let $R$ be a Hasse-Diagram and $E \subseteq R$.
  	Furthermore, let $\cX$ be the set returned by the function $\textsc{enumeratePredecessors}(E, R)$.
  	Then for every Hasse diagram $Q$ we have  $Q \in \cX$  if and only if $(R \setminus E)^* \subseteq Q^* \subseteq R^*$.
  \end{lemma}
  
  \begin{proof}
  	We prove the lemma by induction on $|E|$.
  	If $E = \emptyset$, then $\cX = \{ R \}$ and the statement holds.
  	
  	Otherwise, let $(u, v)$ be the edge selected in line 4 of the algorithm and define $E_1 = E \setminus \{ (u, v) \}$, $R_1 = R \setminus \{ (u, v) \}$, $E_2 = E_1 \cup S'$, and $R_2 = R_1 \cup S'$ (where $S'$ is as in line 7).
  	Let $A$ and $B$ be the results of the recursive calls in lines 5 and 9.
  	
  	First, assume $Q \in \cX = A \cup B$.
  	If $Q \in A$, then by induction we have $(R \setminus E_1)^* \subseteq Q^* \subseteq R^*$.
  	As $R \setminus E \subseteq R \setminus E_1$, it follows that $(R \setminus E)^* \subseteq Q^*$.
  	If $Q \in B$, then by induction we have $(R_2 \setminus E_2)^* \subseteq Q^* \subseteq R_2^*$.
  	Since  $R \setminus E = R_2 \setminus E_2$ and $R_2^* \subseteq R^*$ (by \cref{lem:51}), we have $(R \setminus E)^* \subseteq Q^* \subseteq R^*$.
  	
  	Second, assume $(R \setminus E)^* \subseteq Q^* \subseteq R^*$.
  	If $(u, v) \in Q$, then $(R \setminus E_1)^* \subseteq Q^*$ and, by induction, $Q \in A$.
  	If $(u, v) \notin Q$, then $Q^* \subseteq R^* \setminus \{ (u, v) \} = R_2^*$ and by induction $Q \in B$.
  	The last equality is due to \cref{lem:51} \eqref{pointTwo}.
  \end{proof}

  By \cref{lem:51}~\eqref{pointThree} we obtain the following corollary.
  
  \begin{corollary}
    Let $P = (\universe, R)$ be a poset with Hasse diagram $R_H$ and $(u, v) \in R_H$.
    Moreover, let $\cX$ be the set returned by $\textsc{enumeratePredecessors}(\{ (u, v) \}, R_H)$. Then for every Hasse diagram $Q$ we have $Q \in \cX$ if and only if $R = (Q \cup \{(u,v)\})^*$.
  \end{corollary}\vspace{-0.5cm}

\renewcommand{\thefootnote}{\fnsymbol{footnote}}
\begin{table*}
  \centering\begin{threeparttable}
	\small
	\centering
	\caption{Backward search parameters running times and results of our experiments.}
	\label{tab:prameters_times}
	\begin{tabular}{lll|ll|lrrrr}
		& & & efficiency & full & & & \multicolumn{3}{c}{running time} \\
		$n$ & $C(n)$ & eff.\ $\effTot$  & bandwidth                                                                           & layers & sortable & posets\tnote{$\dagger$}                                                                                                                                                                                                                                                                               & forward & backward                                                                                                              & total   \\\hline
		11  & 26     & $0.59$ & $0.05$                                                                              & $4$    & yes      & $3.7\cdot 10^6$                                                                                                                                                                                                                                                                      & 15\,s   & 111\,ms                                                                                                               & 15\,s   \\
		12  & 29     & $0.89$ & $0.05$                                                                              & $5$    & no       & $6190$\tnote{$\ddagger$} & 15\,ms & 137\,ms & 152\,ms \\
		13  & 33     & $0.72$ & $0.05$                                                                              & $6$    & no       & $516032$                                                                                                                                                                                                                                                                             & 1.7\,s  & 0.6\,s                                                                                                                & 2.3\,s  \\
		13  & 33     & $0.72$ & $0.08$                                                                              & $6$    & no       & $254574$                                                                                                                                                                                                                                                                             & 1.0\,s  & 1.4\,s                                                                                                                & 2.4\,s  \\
		14  & 37     & $0.63$ & $0.12$                                                                              & $9$    & no       & $8.3\cdot 10^6$                                                                                                                                                                                                                                                                      & 32\,s   & 24\,s                                                                                                                 & 56\,s   \\
		15  & 41     & $0.59$ & $0.15$                                                                              & $10$   & no       & $137\cdot 10^6$                                                                                                                                                                                                                                                                      & 11.5\,m   & 6.7\,m                                                                                                                  & 17.2\,m   \\
		16  & 45     & $0.59$ & $0.2$                                                                               & $11$   & no       & $3.99\cdot 10^9$                                                                                                                                                                                                                                                                     & 3\,h    & 8\,h                                                                                                                  & 11\,h   \\
		17  & 49     & $0.63$ & $0.24$                                                                              & $14$   & no       & $13.2\cdot 10^9$                                                                                                                                                                                                                                                                     & 66\,h   & 49\,h                                                                                                                 & 115\,h  \\
		18  & 53     & $0.71$ & $0.19$\tnote{$\dagger\dagger$} & $13$ & no & $10.0\cdot 10^9$ & 56\,h & 33\,h & 89\,h \\
		19  & 57     & $0.84$ & $0.01$                                                                              & $8$    & no       & $318\cdot 10^6$                                                                                                                                                                                                                                                                      & 26\,m   & 135\,s                                                                                                                & 29\,m   \\
		22  & 70     & $0.95$ & $0$                                                                                 &        & no       & $579813$                                                                                                                                                                                                                                                                             & 85\,s   &                                                                                                                       & 85\,s   \\
		28  & 98     & $0.96$ & $0$                                                                                 &        & no       & $30.2\cdot 10^6$                                                                                                                                                                                                                                                                     & 10\,h   &                                                                                                                       & 10\,h\\\hline
	\end{tabular}
\vspace{.5mm} 
  \scriptsize
  \begin{tablenotes}
    \item[$\dagger$] \scriptsize Posets refers to the number of posets stored during both the backward search and the forward search.
    It is less than the total number of posets explored, as some are not stored.
    It is more than the number of posets stored at the end of the execution, as the forward search can discard posets already in storage due to memory constraints.
    \item[$\ddagger$] \scriptsize The number of posets explored for $n=12$ is significantly larger than the numbers given in~\cite{Knuth,peczarski2004new}. In this case the bidirectional search increased the search space: most of the posets can be attributed to the backward search.
    \item[$\dagger\dagger$] \scriptsize  For layers 36 to 40 we used an efficiency bandwidth of $0.23$.
    \end{tablenotes}
  \end{threeparttable}
\end{table*}

\begin{table*}
	\centering~\caption{Running times of our and previous algorithms. The running times due to Peczarski are taken from~\cite{peczarski2012towards}. The year  in the first row indicates the year of publication;
    the year in parentheses indicates the year when the computation was performed.
		For running times for KSI (Kasai, Sawato, Iwata) see~\cite{kasai1994thirty}, and for Wells see~\cite{wells1965applications}.
	}\label{tab:comp_prev_work}
\centering\begin{threeparttable}
	\vspace{1mm}
	{\small
		\begin{tabular}{lll|rrrrrrr}
			& & & Wells & KSI & \multicolumn{4}{c}{\!\!\!Peczarski} & This Work \\ 
			 &  &  & 1965  & 1994   & 2002      & 2004   & 2007     & 2012                   & 2023                  \\ 
             &  &  & \scriptsize (1964)  &    & \scriptsize(2002)      &\scriptsize (2003)   & \scriptsize (2004)     & \scriptsize (2006)      & \scriptsize (2022)                  \\ 
            $n$ & $C(n)$ & $S(n)$        &       &        &           &        &          & \scriptsize 2Threads\! & \scriptsize 48Threads \\\hline
			12  & 29     & 30     & 60\,h & 7\,m   &           &        &          &                        & 152\,ms               \\
			13  & 33     & 34     &       & 230\,h & 10{.}5\,h & 41\,m  & 10\,m    & 46\,s                  & 2.3\,s                \\
			14  & 37     & 38     &       &        &           & 391\,h & 44\,h    & 4{.}5\,h               & 56\,s                 \\
			15  & 41     & 42     &       &        &           &        & 17554\,h\tnote{$\dagger$} &                        & 17.2\,m       
		\end{tabular}
	}
\scriptsize
\begin{tablenotes}
	\item[$\dagger$] \scriptsize 
	The experiments were carried out on up to 24 processors.
\end{tablenotes}
\end{threeparttable}
\end{table*}

  \section{Experiments}\label{sec:experiments}

  We ran our algorithm on a machine with two Intel(R) Xeon(R) E5-2650v4 CPUs (2.20\,GHz, 12 Cores/24 Threads, 30\,MB L3-Cache per CPU), 768\,GB of RAM, 4\,TB NVME SSD storage and 12\,TB hard disk storage.
  We used GNU's \texttt{g++} (10.3.0) optimized with flags \texttt{-O3 -march=native -flto} on Ubuntu 21.04.
  The source code can be found in \url{https://github.com/CodeCrafter47/sortinglowerbounds}.
  It is based on an implementation of Peczarski's algorithm developed by Julian Obst as part of his bachelor's thesis~\cite{obst2019experimentelle}.

 In \cref{fig:params15} we can see the effect of changing the efficiency bandwidth on the running time and the number of explored posets for $n=15$. 
 We can observe that a minimal number of explored posets does not necessarily mean a minimal running time. 
 This is due to the effect that exploring a poset during the backward search is more expensive than during the forward search.

\begin{figure}
	\centering
	\includegraphics[width=.48\textwidth]{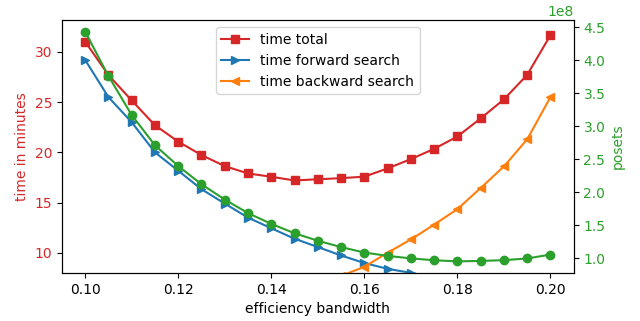}
	\caption{Influence of efficiency bandwidth on running time and size of search space. Here $n=15$ and there are $10$ full layers for the backward search.}
	\label{fig:params15}
\end{figure}

  In \cref{tab:prameters_times} we give details on our experiments including the parameters used for the partial backward search, \ie the number of full layers and the efficiency bandwidth.
  Recomputing the known values of $S(n)$ for $n \in \{11, 12, 13, 14, 15, 19, 22\}$ confirms the results obtained in previous work.
  A comparison of our running times to previous work is shown in \cref{tab:comp_prev_work}.
  Be aware that our program is not optimized for solving small instances fast; instead we aim to solve the larger instances $n = 16,17,18$. 
  Nevertheless, for $n=14 $ we can observe a speedup of roughly 350x compared to Peczarski's latest implementation (using only 24x as many threads).
  This highlights the enormous improvement obtained through the bidirectional search.

  \subsection*{Conclusion.}

  We have computed the minimum number of comparisons required to sort $n$ elements for $n=16, 17, 18$ and provided a new lower bound for $n=28$.
  For $n\geq23$ the value of $S(n)$ is still a wide open questions.
  We believe that some further improvement of the current methods allows for a computation of $S(23)$ and $S(24)$.
  However, for larger values of $n$ other methods (not relying on pure computational power) seem to be necessary.
  On the other hand, it might be possible to actually find new sorting algorithms beating the Ford-Johnson algorithm using an exhaustive computer search with similar methods as in our work.

  \subsection*{Acknowledgements.}
We want to thank Stefan Funke (FMI, University of Stuttgart) for providing us the necessary hardware.

  \bibliography{main}


\FloatBarrier
  \clearpage
  \appendix

  \FloatBarrier

\FloatBarrier
  \section{Further Experimental Results}
  \FloatBarrier

  \begin{table}
    \small
    \centering
    \caption{Explored posets when proving that 16 elements are not sortable using 45 comparisons. Due to the parallelization the numbers for the forward search are not deterministic.}
    \begin{tabular}{rrrr}
      \toprule
      $c$ & backward & forward total & forward sortable \\
      \midrule
      0   & 0        & 1             & 0                \\
      1   & 0        & 1             & 0                \\
      2   & 0        & 1             & 0                \\
      3   & 0        & 1             & 0                \\
      4   & 0        & 3             & 0                \\
      5   & 0        & 6             & 0                \\
      6   & 0        & 13            & 0                \\
      7   & 0        & 28            & 0                \\
      8   & 0        & 79            & 0                \\
      9   & 0        & 216           & 0                \\
      10  & 0        & 536           & 0                \\
      11  & 0        & 1332          & 0                \\
      12  & 7        & 3419          & 0                \\
      13  & 27       & 8779          & 0                \\
      14  & 120      & 22647         & 0                \\
      15  & 547      & 58355         & 0                \\
      16  & 2799     & 144522        & 1                \\
      17  & 13111    & 336757        & 7                \\
      18  & 53976    & 734822        & 17               \\
      19  & 200777   & 1478536       & 50               \\
      20  & 650967   & 2714169       & 159              \\
      21  & 1806016  & 4477006       & 459              \\
      22  & 4252380  & 6597314       & 1342             \\
      23  & 8792382  & 8625146       & 4694             \\
      24  & 16147302 & 9982658       & 22712            \\
      25  & 26159419 & 10370026      & 154359           \\
      26  & 37357419 & 10830095      & 1130279          \\
      27  & 46683444 & 17750837      & 6715593          \\
      28  & 49479591 & 54072329      & 29289077         \\
      29  & 43134486 & 183457116     & 96257441         \\
      30  & 31183260 & 523435272     & 227013407        \\
      31  & 19016824 & 1224756072    & 356223429        \\
      32  & 10080033 & 823309419     & 241783633        \\
      33  & 4674504  & 790331807     & 142762933        \\
      34  & 1956027  & not counted   & not counted      \\
      35  & 13128260 & 0             & 0                \\
      36  & 3148032  & 0             & 0                \\
      37  & 748110   & 0             & 0                \\
      38  & 175269   & 0             & 0                \\
      39  & 40179    & 0             & 0                \\
      40  & 8888     & 0             & 0                \\
      41  & 1867     & 0             & 0                \\
      42  & 358      & 0             & 0                \\
      43  & 63       & 0             & 0                \\
      44  & 8        & 0             & 0                \\
      45  & 1        & 0             & 0                \\
      \bottomrule
    \end{tabular}
    \label{tab:numbers_16}
  \end{table}

  \begin{table}
    \small
    \caption{Explored posets when proving that 17 elements are not sortable using 49 comparisons. Due to the parallelization the numbers for the forward search are not deterministic.}
    \centering
    \begin{tabular}{rrrr}
      \toprule
      \textbf{c} & \textbf{backward} & \textbf{forward} & \textbf{forward sortable} \\
      \midrule
      0          & 0                 & 1                & 0                         \\
      1          & 0                 & 1                & 0                         \\
      2          & 0                 & 1                & 0                         \\
      3          & 0                 & 1                & 0                         \\
      4          & 0                 & 2                & 0                         \\
      5          & 0                 & 3                & 0                         \\
      6          & 0                 & 6                & 0                         \\
      7          & 0                 & 13               & 0                         \\
      8          & 0                 & 38               & 0                         \\
      9          & 0                 & 93               & 0                         \\
      10         & 0                 & 191              & 0                         \\
      11         & 0                 & 379              & 0                         \\
      12         & 1                 & 836              & 0                         \\
      13         & 7                 & 1880             & 0                         \\
      14         & 51                & 4171             & 0                         \\
      15         & 345               & 8957             & 0                         \\
      16         & 1832              & 18845            & 0                         \\
      17         & 8994              & 38808            & 0                         \\
      18         & 41488             & 76896            & 0                         \\
      19         & 171925            & 142572           & 1                         \\
      20         & 658099            & 238328           & 7                         \\
      21         & 2263013           & 354126           & 20                        \\
      22         & 6777866           & 461872           & 67                        \\
      23         & 17955445          & 527787           & 176                       \\
      24         & 41765909          & 529913           & 476                       \\
      25         & 84987026          & 471159           & 1615                      \\
      26         & 152270879         & 383322           & 9504                      \\
      27         & 244370163         & 389412           & 80550                     \\
      28         & 349387637         & 1309265          & 672811                    \\
      29         & 434085076         & 8635871          & 4926106                   \\
      30         & 458231170         & 52049816         & 29625773                  \\
      31         & 403251176         & 248754756        & 136635748                 \\
      32         & 291323005         & 908932076        & 459972165                 \\
      33         & 173049756         & 2111853697       & 946601985                 \\
      34         & 84202562          & 5020719986       & 1564414587                \\
      35         & 34160719          & not counted      & not counted               \\
      36         & 1612229393        & 0                & 0                         \\
      37         & 377447832         & 0                & 0                         \\
      38         & 88042972          & 0                & 0                         \\
      39         & 20471433          & 0                & 0                         \\
      40         & 4737483           & 0                & 0                         \\
      41         & 1087856           & 0                & 0                         \\
      42         & 246083            & 0                & 0                         \\
      43         & 54428             & 0                & 0                         \\
      44         & 11569             & 0                & 0                         \\
      45         & 2333              & 0                & 0                         \\
      46         & 424               & 0                & 0                         \\
      47         & 71                & 0                & 0                         \\
      48         & 8                 & 0                & 0                         \\
      49         & 1                 & 0                & 0                         \\
      \bottomrule
    \end{tabular}
    \label{tab:numbers_17}
  \end{table}

  \begin{table}
    \small
    \caption{Explored posets when proving that 18 elements are not sortable using 53 comparisons. Due to the parallelization the numbers for the forward search are not deterministic.}
    \centering
    \begin{tabular}{rrrr}
      \toprule
      \textbf{c} & \textbf{backward} & \textbf{forward} & \textbf{forward sortable} \\
      \midrule
      0          & 0                 & 1                & 0                         \\
      1          & 0                 & 1                & 0                         \\
      2          & 0                 & 1                & 0                         \\
      3          & 0                 & 1                & 0                         \\
      4          & 0                 & 2                & 0                         \\
      5          & 0                 & 3                & 0                         \\
      6          & 0                 & 5                & 0                         \\
      7          & 0                 & 11               & 0                         \\
      8          & 0                 & 26               & 0                         \\
      9          & 0                 & 45               & 0                         \\
      10         & 0                 & 75               & 0                         \\
      11         & 0                 & 144              & 0                         \\
      12         & 0                 & 300              & 0                         \\
      13         & 0                 & 620              & 0                         \\
      14         & 0                 & 1241             & 0                         \\
      15         & 4                 & 2451             & 0                         \\
      16         & 26                & 5176             & 0                         \\
      17         & 145               & 11135            & 0                         \\
      18         & 675               & 22701            & 0                         \\
      19         & 2894              & 41276            & 0                         \\
      20         & 11282             & 66216            & 0                         \\
      21         & 44841             & 92701            & 0                         \\
      22         & 172993            & 113392           & 0                         \\
      23         & 634769            & 121822           & 1                         \\
      24         & 2214659           & 114568           & 3                         \\
      25         & 6917730           & 94481            & 4                         \\
      26         & 18341765          & 68769            & 7                         \\
      27         & 41133798          & 44673            & 12                        \\
      28         & 80817650          & 26114            & 24                        \\
      29         & 142938247         & 14178            & 190                       \\
      30         & 226165662         & 11464            & 3725                      \\
      31         & 316477949         & 77392            & 60097                     \\
      32         & 388014312         & 850815           & 598475                    \\
      33         & 413123474         & 6817719          & 4285187                   \\
      34         & 380426702         & 40657441         & 23380168                  \\
      35         & 304032816         & 183261692        & 97325918                  \\
      36         & 525993822         & 573715259        & 296751962                 \\
      37         & 304093842         & 1321568879       & 565651499                 \\
      38         & 156193137         & 2356314839       & 889815226                 \\
      39         & 69196030          & 2758245207       & 923822493                 \\
      40         & 26609971          & not counted      & not counted               \\
      41         & 607131408         & 0                & 0                         \\
      42         & 136945779         & 0                & 0                         \\
      43         & 30862867          & 0                & 0                         \\
      44         & 6930638           & 0                & 0                         \\
      45         & 1543827           & 0                & 0                         \\
      46         & 338555            & 0                & 0                         \\
      47         & 72417             & 0                & 0                         \\
      48         & 14872             & 0                & 0                         \\
      49         & 2881              & 0                & 0                         \\
      50         & 506               & 0                & 0                         \\
      51         & 80                & 0                & 0                         \\
      52         & 9                 & 0                & 0                         \\
      53         & 1                 & 0                & 0                         \\
      \bottomrule
    \end{tabular}
    \label{tab:numbers_18}
  \end{table}

  \begin{table}[]
    \small
    \centering
    \caption{Summary of known values for $S(n)$ and new results.
    Here $C(n)$ is the information theoretic lower bound, $F(n)$ is the number of comparisons required by the Ford-Johnson algorithm and $S(n)$ is the minimum number of comparisons required to sort $n$ elements.
    Lines marked with a star contain new results obtained in this paper.}
    \label{tab:results}
    \begin{tabular}{lllll}
      \toprule
      $n$ & $C(n)$ & $F(n)$ & $S(n)$                                         \\
      \midrule
      1   & 0      & 0      & 0                                              \\
      2   & 1      & 1      & 1                                              \\
      3   & 3      & 3      & 3                                              \\
      4   & 5      & 5      & 5                                              \\
      5   & 7      & 7      & 7                                              \\
      6   & 10     & 10     & 10                                             \\
      7   & 13     & 13     & 13                                             \\
      8   & 16     & 16     & 16                                             \\
      9   & 19     & 19     & 19                                             \\
      10  & 22     & 22     & 22                                             \\
      11  & 26     & 26     & 26                                             \\
      12  & 29     & 30     & 30~\cite{wells2014elements}                    \\
      13  & 33     & 34     & 34~\cite{kasai1994thirty,peczarski2002sorting} \\
      14  & 37     & 38     & 38~\cite{peczarski2004new}                     \\
      15  & 41     & 42     & 42~\cite{chinese2007,peczarski2007ford}        \\
      16  & 45     & 46     & 46~$\bigstar$                                  \\
      17  & 49     & 50     & 50~$\bigstar$                                  \\
      18  & 53     & 54     & 54~$\bigstar$                                  \\
      19  & 57     & 58     & 58~\cite{chinese2007}                          \\
      20  & 62     & 62     & 62                                             \\
      21  & 66     & 66     & 66                                             \\
      22  & 70     & 71     & 71~\cite{peczarski2004new}                     \\
      23  & 75     & 76     & 75 or 76                                       \\
      24  & 80     & 81     & 80 or 81                                       \\
      25  & 84     & 86     & 84, 85 or 86                                   \\
      26  & 89     & 91     & 89, 90 or 91                                   \\
      27  & 94     & 96     & 94, 95 or 96                                   \\
      28  & 98     & 101    & 99, 100 or 101~$\bigstar$                      \\
      29  & 103    & 106    & 103, 104, 105 or 106                           \\
      30  & 108    & 111    & 108, 109, 110 or 111                           \\
      31  & 113    & 116    & 113, 114, 115 or 116                           \\
      32  & 118    & 121    & 118, 119, 120 or 121                           \\
      33  & 123    & 126    & 123, 124, 125 or 126                           \\
      34  & 128    & 131    & 128, 129, 130 or 131                           \\
      35  & 133    & 136    & 133, 134, 135 or 136                           \\
      36  & 139    & 141    & 139, 140 or 141                                \\
      37  & 144    & 146    & 144, 145 or 146                                \\
      38  & 149    & 151    & 149, 150 or 151                                \\
      39  & 154    & 156    & 154, 155 or 156                                \\
      40  & 160    & 161    & 160 or 161                                     \\
      41  & 165    & 166    & 165 or 166                                     \\
      42  & 170    & 171    & 170 or 171                                     \\
      43  & 176    & 177    & 176 or 177                                     \\
      44  & 181    & 183    & 181, 182 or 183                                \\
      45  & 187    & 189    & 187, 188 or 189                                \\
      46  & 192    & 195    & 192, 193, 194 or 195                           \\
      47  & 198    & 201    & 198, 199 or 200~\cite{SCHULTEMONTING198119}    \\
      \bottomrule
    \end{tabular}
  \end{table}
\end{document}